\documentclass[a4paper,12pt]{article}
\usepackage{bm, color}
\usepackage{amsmath,amsfonts,amssymb}
\usepackage{amsthm}
\usepackage{tikz}
\usetikzlibrary{calc}
\usepackage{placeins}
\usepackage{subfigure}
\usepackage{graphicx}
\usepackage{amsmath,amsfonts,amssymb,bm}
\usepackage{color}
\newcommand{\te}{\mathrm{t}}

\newcommand{\bk}{\mathbf{k}}
\newcommand{\eps}{\varepsilon}
\newcommand{\R}{\mathbb{R}}
\newcommand{\sv}{\, ,}

\newcommand{\eq}[1]{(\ref{#1})}
\newcommand*{\diff}{\mathop{}\!\mathrm{d}}

\newtheorem{definition}{Definition}

\newtheorem{lemma}{Lemma}
\newtheorem{theorem}{Theorem}
\newtheorem{remark}{Remark}
\newtheorem{corollary}{Corollary}
\begin{document}
\title{Mathematical aspects and simulation of electron-electron scattering in graphene}
\date{}
\author{ Giovanni Nastasi\thanks{Department of Mathematics and Computer Science, University of Catania, Viale Andrea Doria 6, 
95125 Catania, Italy  ({\tt giovanni.nastasi@unict.it}).} 
\and Vittorio Romano\thanks{Department of Mathematics and Computer Science, University of Catania, Viale Andrea Doria 6, 
95125 Catania, Italy  ({\tt romano@dmi.unict.it}).} }
\maketitle

\begin{abstract}
Some properties of the electron-electron collision operator in graphene are analyzed along with the evaluation of collision rate. Monte Carlo simulations complete the study and highlight  the non negligible role of the electron-electron scattering for an accurate evaluation of the currents and, as a consequence,  of the characteristic curves.  
\end{abstract}


\section{Introduction}
The discovery that free-standing graphene indeed there exists \cite{ART:Nov1,ART:Nov2} behind the theoretical prediction has paved the way for a huge literature on 2D crystals.  Graphene has interesting electric features \cite{ART:GeNo,ART:CaNe} but it is also a challenging material from an electro-magneto-mechanical \cite{ART:Sfyris} and thermal point of view \cite{ART:CaTo}. In this paper
some mathematical and computational issues  arising when the electron-electron scattering is included in the analysis of charge transport in graphene are investigated. The electron-phonon interactions are more effective but electron-electron scattering cannot be neglected in some regimes.  The appropriate expression for the electron-electron scattering rate is still matter of debate (see for example \cite{ART:CastroNeto,ART:Tomadin,ART:Li,ART:Fang,ART:Sano,ART:Borowik,ART:Hwang,ART:DasSarma,ART:Brida}); however, some mathematical considerations can be made only on the basis of the general properties, such conservation of energy and crystal momentum during the collision when umklapp effects are neglected. These can be relevant but usually at low temperatures. Here the analysis is performed having in mind room temperature.

We have found the general form of the collisional invariants and the kernel of the 
 electron-electron scattering. The latter is constituted by a family of distribution functions which generalize the Fermi-Dirac one and belong to a five-dimensional manifold. If both electron-phonon and electron-electron scatterings are taken into account, the equilibria lie on a two-dimensional sub-manifold.   It is interesting to stress that the results are analogous to that obtained in the relativistic case \cite{Cercignani}. This is not surprising since electrons inside graphene move like massless fermions for wave-vectors around the Dirac points \cite{Katnelson}. However, our derivation requires less regularity than that assumed in \cite{Cercignani}  and does not use concepts of Special Relativity. 
 
 To complete the investigation, an estimation of the electron-electron scattering rate is given and employed in a DSMC code based on the approach proposed in \cite{RoMajCo} 
 which allows an efficient and accurate inclusion of the Pauli exclusion principle. Simulations of charge transport in suspended monolayer graphene indicate a variation in the 
 average electron velocity of about 13\%, at least in the considered cases, giving a confirmation of the importance of the electron-electron scattering according to \cite{ART:Li}. 
 
 The plan of the paper is as follows. In Sec. 2 the semiclassical Boltzmann equation is introduced and in particular the main features of the electron-electron scattering are recalled.  Sec. 3 is devoted to the analysis of the  properties of the electron-electron collision operator in graphene such as the collision invariants and the kernel. In Sec.s 4 and 5 the electron-electron scattering rate is evaluated  in the intra-band case and in the inter-band case, respectively. The last section contains the results of the Monte Carlo simulations in order to assess the importance of the electron-electron scattering with respect to the electron-phonon ones.

\section{Semiclassical transport and electron-electron scattering}
Electrons moving inside a graphene sheet are located in the wave-vector space mainly around the Dirac points (here considered as equivalent). With a good approximation, the energy bands can be assumed linear in the modulus of the wave-vector   
$\mathbf{k}$, measured with respect to the Dirac point $\mathbf{k}_l$,  with a zero gap between the valence and the conduction bands. Therefore, we assume that the energy band $\varepsilon$ with respect to the Dirac energy $\varepsilon_D$ is given by  
\begin{equation*}
\varepsilon(\mathbf{k}) = \alpha\hbar v_F \vert \mathbf{k}  -  \mathbf{k}_l\vert,
\end{equation*}
with ${\alpha} = 1$ in the conduction band and ${\alpha} = -1$ in the valence band, 
and extend such a dispersion relation to  any $\mathbf{k} \in \R^2$.

The electron velocity is given by  
\begin{equation*}
\mathbf{v}_{\alpha} = \frac{1}{\hbar}\nabla_{\mathbf{k}}\varepsilon_{\alpha}.
\end{equation*}
Observe that the energy band is negative in the valence band. 
 
The semiclassical Boltzmann equations for the distributions (occupation numbers) $f_\alpha(\mathbf{r},\mathbf{k},t)$, at the position $\mathbf{r}$, wave-vector $\mathbf{k}$ and time $t$,  of electrons in the conduction band ($\alpha = 1$) and in the valence band  ($\alpha = -1$) read (see \cite{ART:Tomadin})
\begin{align}
&\frac{\partial f_\alpha(\mathbf{k})}{\partial t}+\mathbf{v}_{\alpha}\cdot\nabla_{\mathbf{r}}f_\alpha(\mathbf{k})-\frac{e}{\hbar}\mathbf{E}\cdot\nabla_{\mathbf{k}}f_\alpha(\mathbf{k})\nonumber\\
&= C[f_\alpha] := \sum_{\nu,\alpha',\mathbf{k}'}\left[ S_{\alpha',\alpha}^{\nu}(\mathbf{k}',\mathbf{k}) f_{\alpha'}(\mathbf{k}') (1-f_{\alpha}(\mathbf{k})) - S_{\alpha,\alpha'}^{\nu}(\mathbf{k},\mathbf{k'}) f_{\alpha}(\mathbf{k}) (1-f_{\alpha'}(\mathbf{k}')) \right]\nonumber\\
&+\sum_{\alpha',\mathbf{k}_*, \mathbf{k}', \mathbf{k}'_*} \sum_{\beta \in \{\alpha', - \alpha' \}}\left[ S_{ee} (\mathbf{k}',\mathbf{k}'_*,\mathbf{k},\mathbf{k}_*) f_{\alpha}(\mathbf{k}')f_{\alpha'}(\mathbf{k}'_*)  (1-f_\alpha(\mathbf{k}))(1-f_{\beta}(\mathbf{k}_*))\right. \nonumber\\
&\left. - S_{ee}(\mathbf{k},\mathbf{k}_*,\mathbf{k}',\mathbf{k}'_*) f_\alpha(\mathbf{k})f_{\alpha'}(\mathbf{k}_*) (1-f_{\alpha}(\mathbf{k}'))(1-f_{\beta}(\mathbf{k}'_*)) \right] \delta_{\mathbf{k}_1 + \mathbf{k}_2 = \mathbf{k}'_1 + \mathbf{k}'_2}, \label{transport}
\end{align}
where the Kronecker symbol $\delta_{\mathbf{k}_1 + \mathbf{k}_2 = \mathbf{k}'_1 + \mathbf{k}'_2}
$ expresses the momentum conservation.
We have explicitly written only the dependence of  $f_\alpha$ on $\mathbf{k}$ for the sake of simplicity. 
Eq.s (\ref{transport}) include both the electron-phonon and the electron-electron scatterings, even if for low densities the second one is usually considered negligible.

The $S_{\alpha',\alpha}^{\nu}$'s are the transition rate due to the scattering of electrons with phonons in the $\nu$ branch, where $\nu$ ranges in the set LA (longitudinal acoustic), 
TA (transversal acoustic), LO (longitudinal optical), TO (transversal optical) and $K$ phonons. Indeed there exist also the flexural modes, the so-called $Z$-phonons, but they do not interact 
with electrons and play a relevant role only for the determination of the crystal temperature \cite{CoRo}. $S_{ee}$ is the electron-electron transition rate.  Umklapp effects will be neglected.

The transition rate due to the scattering of electrons with phonons in the $\nu$th branch has the form
\begin{eqnarray}
&&
S_{\alpha', \alpha}(\bk', \bk) =
  \left| G^{(\nu)}_{\alpha', \alpha}(\bk', \bk) \right|^{2} \left[ 
\left( n^{(\nu)}_{\mathbf{q}} + 1 \right)
\delta \left( \eps_{\alpha}(\bk) - \eps_{\alpha'}(\bk') + \hbar \, 
\omega^{(\nu)}_{\mathbf{q}}
\right) \right. \nonumber
\\
&&
\left. \mbox{} \hspace{137pt} + n^{(\nu)}_{\mathbf{q}} \,
\delta \left( \eps_{\alpha}(\bk) - \eps_{\alpha'}(\bk') - \hbar \, 
\omega^{(\nu)}_{\mathbf{q}}
\right) \right] .  \label{Scatt}
\end{eqnarray}
 $\left| G^{(\nu)}_{\alpha', \alpha}(\bk', \bk) \right|$ is the 
matrix element,  the symbol $\delta$ denotes the Dirac distribution, $\omega^{(\nu)}_{\mathbf{q}}$ is the
the $\nu$th  phonon  frequency, $n^{(\nu)}_{\mathbf{q}}$ is the Bose-Einstein distribution for 
the phonon of type $\nu$
$$
n^{(\nu)}_{\mathbf{q}} = \dfrac{1}{e^{\hbar \, \omega^{(\nu)}_{\mathbf{q}} /k_B T} - 1} \sv
$$
$k_B$ is the Boltzmann constant and $T$ is the  graphene lattice temperature which, in this article, will be kept constant.
When, for a phonon $\nu_{*}$, $\hbar \, \omega^{(\nu_{*})}_{\mathbf{q}} \ll k_B T$, then the 
scattering with the phonon $\nu_{*}$ can be assumed elastic. 
In this case, we eliminate in Eq.~\eq{Scatt} the term 
$\hbar \, \omega^{(\nu_{*})}_{\mathbf{q}}$ inside the delta distribution and we use the 
approximation $n^{(\nu_{*})}_{\mathbf{q}} + 1 \approx n^{(\nu_{*})}_{\mathbf{q}}$. \\
For acoustic phonons, usually one considers the elastic approximation, and
\begin{equation}
2 \, n^{(ac)}_{\mathbf{q}}
\left| G^{(ac)}(\bk', \bk) \right|^{2} = 
\dfrac{\pi \, D_{ac}^{2} \, k_{B} \, T}{4 \hbar \, \sigma_m \, v_{p}^{2}}
\left( 1 + \cos \vartheta_{\bk \sv \bk'} \right) ,
\label{transport_acoustic}
\end{equation}
where $D_{ac}$ is the acoustic phonon coupling constant, $v_{p}$ is the sound speed in 
graphene, 
$\sigma_m$ the graphene areal density, 
and $\vartheta_{\bk \sv \bk'}$ is the convex angle between $\bk$ and ${\bk'}$.   
\\
There are three relevant optical phonon scatterings: the longitudinal optical (LO), the 
transversal optical (TO) and the ${K}$  phonons.
The matrix elements are
\begin{align}
&
\left| G^{(LO)}(\bk', \bk) \right|^{2} + \left| G^{(TO)}(\bk', \bk) \right|^{2} =
 \dfrac{\pi \, D_{O}^{2}}{\sigma_m \, \omega_{O}}
\\
&
\left| G^{(K)}(\bk', \bk) \right|^{2} = 
\dfrac{ \pi \, D_{K}^{2}}{\sigma_m \, \omega_{K}}
\left( 1 - \cos \vartheta_{\bk \sv \bk'} \right) ,
\end{align}
where $D_{O}$ is the optical phonon coupling constant, $\omega_{O}$ the optical phonon 
frequency, $D_{K}$ is the K-phonon coupling constant and $\omega_{K}$ the K-phonon frequency.
Physical parameters for the collision terms are summarized in Table \ref{tabella}.
\begin{table}[!ht]
\centering
\begin{tabular}{||l|c||l|c||}
\hline & & & \\[-10pt]
$v_{F}$ & $10^{8}$ cm/s & $v_{p}$ & $2 \times 10^{6}$ cm/s \\[2pt]
\hline & & & \\[-10pt]
$\sigma_{m}$ & $7.6 \times 10^{-8}$ g/cm$^{2}$ & $D_{ac}$ & $6.8$ eV \\[2pt]
\hline & & & \\[-10pt]
$\hbar \, \omega_{O}$ & $164.6$ meV & $D_{O}$ & $10^{9}$ eV/cm \\[2pt]
\hline & & & \\[-10pt]
$\hbar \, \omega_{K}$ & $124$ meV & $D_{K}$ & $3.5 \times 10^{8}$ eV/cm \\[2pt]
\hline
\end{tabular}
\caption{Physical parameters for the collision terms.}
\label{tabella}
\end{table}

$S_{ee}(\mathbf{k}_1,\mathbf{k}_2,\mathbf{k}'_1,\mathbf{k}'_2)$ represents the probability of transition, for unit time,  from the joint state  
$\left(\mathbf{k}_1, \mathbf{k}_2 \right)$ to the joint state $\left(\mathbf{k}'_1,\mathbf{k}'_2\right)$. In this paper only normal scatterings are considered and umklapp processes  are neglected. 

According to the Fermi golden rule \cite{ART:Li} (for similar approaches in other contexts see \cite{Morandi1,Morandi2}),  if  the ingoing electrons are in the bands $\alpha$, $\alpha'$ and the outgoing electrons  are in the bands $\beta$, $\beta'$, $S_{ee}$ is given by
\begin{equation*}
S_{ee}(\mathbf{k}_1,\mathbf{k}_2,\mathbf{k}'_1,\mathbf{k}'_2) = \frac{2\pi}{\hbar}\vert M \vert^2\delta(\varepsilon_{\beta}(\mathbf{k}'_1) +\varepsilon_{\beta'}(\mathbf{k}'_2) - \varepsilon_{\alpha}(\mathbf{k}_1) -\varepsilon_{\alpha'}(\mathbf{k}_2)).
\end{equation*}
$M$ is the matrix interaction, whose generic element is given by 
\begin{equation*}
\vert M\vert^2 = \frac{1}{2}\left[ \vert V(q)\vert^2+\vert V(q')\vert^2- V(q)V(q') \right],
\end{equation*}
where $V(q)$ is the  Coulomb potential between pairs of electrons having momenta $\mathbf{k}_1$ and $\mathbf{k}'_1$
\begin{align*}
V(q) & = \frac{2\pi e^2}{\epsilon(q)qA}\frac{1+\cos(\phi_{\mathbf{k}_1,\mathbf{k}'_1})}{2}\frac{1+\cos(\phi_{\mathbf{k}_2,\mathbf{k}'_2})}{2},
\end{align*}
with $q=\vert \mathbf{k}_1-\mathbf{k}'_1 \vert$. Similarly
\begin{align*}
V(q') & = \frac{2\pi e^2}{\epsilon(q')q'A}\frac{1+\cos(\phi_{\mathbf{k}_1,\mathbf{k}'_2})}{2}\frac{1+\cos(\phi_{\mathbf{k}_2,\mathbf{k}'_1})}{2},
\end{align*}
with $q'=\vert \mathbf{k}_1-\mathbf{k}'_2 \vert$. 

In the above expressions,  $\phi_{\mathbf{k},\mathbf{k}'}$ denotes the angle between   $\mathbf{k}$ and $\mathbf{k}'$. 
More sophisticated models are present in the literature. The interested reader is referred to \cite{ART:Fang,ART:Sano,ART:Borowik}.

In the \emph{random-phase approximation} (valid when $n\geq 10^{12}$ cm$^{-2}$), the dielectric function $\epsilon(q)$ is given by (see \cite{ART:Hwang})
\begin{equation*}
\epsilon(q) = 1+v_c(q)\Pi(q),
\end{equation*}
where $v_c(q)=2\pi e^2/\kappa q$,  $\kappa$ being the \textit{background lattice dielectric constant} which satisfies
\begin{equation}\label{EQ:rs_eqn}
r_s = \frac{e^2}{\kappa \gamma} \sqrt{\frac{4}{g_s g_v}}.
\end{equation}
Here, $\gamma=\hbar v_F$, $g_s$ and $g_v$ are the spin and valley degenerations respectively;  $r_s$ is a constant representing the adimensional \textit{Wigner-Seitz radius}; $\Pi(q)=D(\varepsilon_F)\widetilde{\Pi}(q)$, with $\varepsilon_F$ the Fermi energy and $D$  is the density of states given by
\begin{equation}\label{EQ:DoS}
D(\varepsilon) = \frac{g_v g_s \vert \varepsilon \vert}{2\pi\gamma^2}.
\end{equation}
The function $\widetilde{\Pi}(q)$ is usually decomposed as  $\widetilde{\Pi}(q)=\widetilde{\Pi}^-(q)+\widetilde{\Pi}^+(q)$ with 
\begin{eqnarray*}
\widetilde{\Pi}^-(q) &=& \frac{\pi q}{8 k_F},\\
\widetilde{\Pi}^+(q) &=& \left\lbrace
\begin{alignedat}{2}
&1-\frac{\pi q}{8 k_F} && \qquad\mbox{if}\quad q<2k_F\\
&1-\frac{\sqrt{q^2-4k_F^2}}{2q}-\frac{q}{4k_F}\arcsin\left( \frac{2k_F}{q} \right) && \qquad\mbox{otherwise}
\end{alignedat}
\right. 
\end{eqnarray*}
where $k_F=\sqrt{4\pi n / g_s g_v}$ is the Fermi momentum which satisfies $\varepsilon_F=\gamma k_F$ and  $n$ is the electron density. More details can be found in  \cite{ART:CastroNeto,ART:DasSarma}.

\section{Properties of the intra-band electron-electron collision operator in  graphene}

In any intra-band binary collision the conservation laws of crystal momentum and energy have to be satisfied in a normal scattering
\cite{ART:Wallace}
\begin{equation}
\hbar\mathbf{k} +\hbar\mathbf{k}_*=\hbar\mathbf{k}'+\hbar\mathbf{k}'_*, \label{conservation1} \end{equation}
\begin{equation}
\hbar v_F\vert\mathbf{k} - \mathbf{k}_l\vert+\hbar v_F\vert\mathbf{k}_* - \mathbf{k}_l\vert=\hbar v_F\vert\mathbf{k}' - \mathbf{k}_l\vert+\hbar v_F\vert\mathbf{k}'_* - \mathbf{k}_l\vert  \label{conservation2}
\end{equation}
 $\mathbf{k}$ and $\mathbf{k}_*$ are the crystal momenta of the incoming particles, $\mathbf{k}'$ and 
$\mathbf{k}'_*$ are the crystal momenta of the outgoing particles. Indeed, in general the crystal momentum is not the momentum of the electron but it is customary in solid state physics to assume the approximation of the conservation of the crystal momentum.

After the translation
$\mathbf{k}-\mathbf{k}_l  \rightarrow \mathbf{k}$
the previous conservation laws read (note that the conservation of momentum is invariant under translation) 
\begin{equation}\label{EQ:cons_mom_2}
\mathbf{k}+\mathbf{k}_*=\mathbf{k}'+\mathbf{k}'_*,
\end{equation}
\begin{equation}\label{EQ:cons_en_2}
\vert\mathbf{k}\vert+\vert\mathbf{k}_*\vert=\vert\mathbf{k}'\vert+\vert\mathbf{k}'_*\vert.
\end{equation}
\begin{definition}
A function $\varphi$ is said a  collisional invariant if it satisfies the condition
\begin{equation}\label{EQ:inv_coll}
\varphi+\varphi_*=\varphi'+\varphi'_*,
\end{equation}
where the following shorthand notation has been used: $\varphi := \varphi(\mathbf{k})$, $\varphi_* := \varphi(\mathbf{k}_*)$, $\varphi' := \varphi(\mathbf{k}')$, $\varphi' _*:= \varphi(\mathbf{k}'_*)$.
\end{definition}
We want to determine the most general form of the collision invariant for the electron-electron scattering (hereafter EES). To this aim we introduce the following two lemmas. 
\begin{lemma}\label{LEMMA:1} Let $ E_n$ be a n-dimensional Euclidean real vector space, endowed with the canonical scalar product 
 and let  $f(\mathbf{x})$ be a real function on $ E_n$  which is continuous  at least in a point and satisfies 
\begin{equation}
f(\mathbf{x})+f(\mathbf{x}_1) = f(\mathbf{x}+\mathbf{x}_1)
\end{equation}
for each $\mathbf{x},\mathbf{x}_1\in E_n$. Then $f(\mathbf{x}) = \mathbf{A}\cdot\mathbf{x}$ where  $\mathbf{A}$ is a constant vector. 
\end{lemma}
The proof is standard. The interested reader can see for example \cite{Cercignani_old}.  
\begin{lemma}\label{LEMMA:2}
Let $f:\mathbb{R}^n\rightarrow\mathbb{R}$ be a differentiable function which is homogeneous of degree one  (that is  $f(\lambda \mathbf{x}) = \lambda f(\mathbf{x})$ $\forall$ $\lambda\in\mathbb{R}$). Then $f$ is linear.
\end{lemma}
\begin{proof}
The homogeneity means
\begin{equation*}
f(\lambda x_1, \ldots, \lambda x_n) = \lambda f(x_1, \ldots, x_n).
\end{equation*}
and by taking the derivatives with respect to $\lambda$ of both sides one has 
\begin{equation*}
f(\mathbf{x}) = \nabla_\mathbf{x} f  (\lambda\mathbf{x}) \cdot \mathbf{x}.
\end{equation*}
Passing to the limit as  $\lambda\rightarrow 0$, one gets
\begin{equation*}
f(\mathbf{x}) = \nabla_\mathbf{x} f \big\vert_{\mathbf{0}} \cdot \mathbf{x}
\end{equation*}
and therefore, being $\nabla_\mathbf{x} f \big\vert_{\mathbf{0}}$ a constant  vector,  it follows that $f(\mathbf{x})$ is  linear.
\end{proof}
Now we are in the position to establish the general form of the collision invariants. 
\begin{theorem}
Let $\varphi (\mathbf{k})$ a continuous function whose odd part is differentiable.  $\varphi (\mathbf{k})$ is a collisional invariant if and only if 
\begin{equation}
\varphi(\mathbf{k})=a+\mathbf{b}\cdot\mathbf{k}+c\vert\mathbf{k}\vert \label{coll_inv}
\end{equation}
for arbitrary $a, c \in \R$ and arbitrary $\mathbf{b} \in \R^2$.
\end{theorem}
\begin{proof}
The equation \eqref{EQ:inv_coll} implies that $\varphi+\varphi_*$ takes the same values for all the pairs of vectors  $(\mathbf{k},\mathbf{k}_*)$ that satisfy the conservation laws (\ref{conservation1}), (\ref{conservation2}), that is  $\varphi+\varphi_*$ is  constant whenever $\mathbf{k}+\mathbf{k}_*$ and $\vert\mathbf{k}\vert+\vert\mathbf{k}_*\vert$ are constant. This means that  $\varphi+\varphi_*$ is a function which depends only on such quantities, in other words 
\begin{equation}\label{EQ:rel1}
\varphi(\mathbf{k})+\varphi(\mathbf{k}_*) = \Phi(\vert\mathbf{k}\vert+\vert\mathbf{k}_*\vert, \mathbf{k}+\mathbf{k}_*)
\end{equation}
for a suitable function $\Phi$.
Let us introduce
\begin{equation*}\label{EQ:rel2}
\varphi_\pm (\mathbf{k}) = \varphi(\mathbf{k}) \pm \varphi(-\mathbf{k}),
\end{equation*}
\begin{equation*}\label{EQ:rel3}
\Phi_\pm(\vert\mathbf{k}\vert+\vert\mathbf{k}_*\vert, \mathbf{k}+\mathbf{k}_*) = \Phi(\vert\mathbf{k}\vert+\vert\mathbf{k}_*\vert, \mathbf{k}+\mathbf{k}_*)\pm\Phi(\vert\mathbf{k}\vert+\vert\mathbf{k}_*\vert, -\mathbf{k}-\mathbf{k}_*).
\end{equation*}
After the substitution $(\mathbf{k},\mathbf{k}_*)\rightarrow (-\mathbf{k},-\mathbf{k}_*)$, relation \eqref{EQ:rel1} becomes
\begin{equation}\label{EQ:rel4}
\varphi(-\mathbf{k})+\varphi(-\mathbf{k}_*) = \Phi(\vert\mathbf{k}\vert+\vert\mathbf{k}_*\vert, -\mathbf{k}-\mathbf{k}_*).
\end{equation}
\eqref{EQ:rel1} $+$ \eqref{EQ:rel4} $\Rightarrow$ $\varphi(\mathbf{k})+\varphi(-\mathbf{k})+\varphi(\mathbf{k}_*)+\varphi(-\mathbf{k}_*) =\Phi(\vert\mathbf{k}\vert+\vert\mathbf{k}_*\vert, \mathbf{k} + \mathbf{k}_*)+ \Phi(\vert\mathbf{k}\vert+\vert\mathbf{k}_*\vert, -\mathbf{k}-\mathbf{k}_*)$.\\
\eqref{EQ:rel1} $-$ \eqref{EQ:rel4} $\Rightarrow$ $\varphi(\mathbf{k})-\varphi(-\mathbf{k})+\varphi(\mathbf{k}_*)-\varphi(-\mathbf{k}_*) =\Phi(\vert\mathbf{k}\vert+\vert\mathbf{k}_*\vert, \mathbf{k} + \mathbf{k}_*)- \Phi(\vert\mathbf{k}\vert+\vert\mathbf{k}_*\vert, -\mathbf{k}-\mathbf{k}_*)$.\\
In a compact form one has 
\begin{equation}\label{EQ:rel5}
\varphi_\pm(\mathbf{k})+\varphi_\pm(\mathbf{k}_*) = \Phi_\pm(\vert\mathbf{k}\vert+\vert\mathbf{k}_*\vert, \mathbf{k}+\mathbf{k}_*).
\end{equation}
Observe that
\begin{equation*}\label{EQ:rel7}
\varphi_\pm(-\mathbf{k})=\pm\varphi_\pm(\mathbf{k}).
\end{equation*}
If we take $\mathbf{k}_*=-\mathbf{k}$  then
\begin{equation}
\varphi_+(\mathbf{k})+\varphi_+(-\mathbf{k}) = 2\varphi_+(\mathbf{k})= \Phi_+(2\vert\mathbf{k}\vert,\mathbf{0}).
\end{equation}
As a consequence  $\varphi_+$ depends only on  $\vert\mathbf{k}\vert$, that is $\varphi_+ (\mathbf{k}) =\psi(\vert\mathbf{k}\vert)$. This implies, on account of  \eqref{EQ:rel5}, that  $\Phi_+$ depends only on $\vert\mathbf{k}\vert+\vert\mathbf{k}_*\vert$ because no function (apart from the constant case) of $\mathbf{k}+\mathbf{k}_*$ can be built from  $\vert\mathbf{k}\vert$ and $\vert\mathbf{k}_*\vert$. In fact, by absurd, let us suppose that $f(\mathbf{k}+\mathbf{k}_*)=g(\vert\mathbf{k}\vert,\vert\mathbf{k}_*\vert)$ then  setting $\mathbf{k}_*=\mathbf{0}$, one has $f(\mathbf{k})=h(\vert\mathbf{k}\vert)$. Therefore $f(\mathbf{k}+\mathbf{k}_*)=h(\vert\mathbf{k}+\mathbf{k}_*\vert)$, that is $h(\sqrt{\vert\mathbf{k}\vert^2+\vert\mathbf{k}_*\vert^2+2\mathbf{k}\cdot\mathbf{k}_*})=g(\vert\mathbf{k}\vert,\vert\mathbf{k}_*\vert)$. By choosing  $\mathbf{k}\cdot\mathbf{k}_*=0$, one gets $h(\sqrt{\vert\mathbf{k}\vert^2+\vert\mathbf{k}_*\vert^2})=g(\vert\mathbf{k}\vert,\vert\mathbf{k}_*\vert)$; instead by choosing $\mathbf{k}_*=\mathbf{k}$, one finds $h(\vert\mathbf{k}\vert+\vert\mathbf{k}_*\vert)=g(\vert\mathbf{k}\vert,\vert\mathbf{k}_*\vert)$. The two expressions of $g$ are compatible only in the cases $\mathbf{k} = {\bf 0}$ or
$\mathbf{k}_*= {\bf 0}$. This implies that  $h$ must be constant leading to a contradiction.  

Now from \eqref{EQ:rel5} it is  possible to write
\begin{equation*}
\psi(\vert\mathbf{k}\vert)+\psi(\vert\mathbf{k}_*\vert)=\Phi_+(\vert\mathbf{k}\vert+\vert\mathbf{k}_*\vert).
\end{equation*}
Setting $\mathbf{k}_*=\mathbf{0}$ we find
\begin{equation*}
\psi(\vert\mathbf{k}\vert)+\psi(0)=\Phi_+(\vert\mathbf{k}\vert)
\end{equation*}
and therefore 
\begin{equation*}
\psi(\vert\mathbf{k}\vert+\vert\mathbf{k}_*\vert)+\psi(0)=\Phi_+(\vert\mathbf{k}\vert+\vert\mathbf{k}_*\vert)
\end{equation*}
wherefrom
\begin{equation*}
\psi(\vert\mathbf{k}\vert)+\psi(\vert\mathbf{k}_*\vert)=\psi(\vert\mathbf{k}\vert+\vert\mathbf{k}_*\vert)+\psi(0).
\end{equation*}
Summing up to both sides the quantity $-2 \psi(0)$ one get
\begin{equation*}
\psi(\vert\mathbf{k}\vert)-\psi(0)+\psi(\vert\mathbf{k}_*\vert)-\psi(0)=\psi(\vert\mathbf{k}\vert+\vert\mathbf{k}_*\vert)-\psi(0).
\end{equation*}
Setting $f(x)=\psi(x)-\psi(0)$ one has
\begin{equation*}
f(\vert\mathbf{k}\vert)+f(\vert\mathbf{k}_*\vert)=f(\vert\mathbf{k}\vert+\vert\mathbf{k}_*\vert).
\end{equation*}
Applying  Lemma \ref{LEMMA:1} in the one dimensional case, we can conclude that $f(\vert\mathbf{k}\vert)=2c\vert\mathbf{k}\vert$. As a consequence
\begin{equation*}
\varphi_+(\vert\mathbf{k}\vert)=\psi(\vert\mathbf{k}\vert)=f(\vert\mathbf{k}\vert)+\psi(0)=2c\vert\mathbf{k}\vert+2a,
\end{equation*}
where $\psi(0)=2a$ is a  constant.

Let us consider the equation \eqref{EQ:rel5} for $\varphi_-$ and observe that
\begin{equation*}
(\mathbf{k}+\mathbf{k}_*)^2 = \vert\mathbf{k}\vert^2+\vert\mathbf{k}_*\vert^2+2\mathbf{k}\cdot\mathbf{k}_*= (\vert\mathbf{k}\vert+\vert\mathbf{k}_*\vert)^2 +2(\mathbf{k}\cdot\mathbf{k}_*-\vert\mathbf{k}\vert\vert\mathbf{k}_*\vert).
\end{equation*}
Then $\Phi_-$ depends only on $\mathbf{k}+\mathbf{k}_*$ if $\mathbf{k}\cdot\mathbf{k}_*=\vert\mathbf{k}\vert\vert\mathbf{k}_*\vert$. In this way one has 
\begin{equation*}
\varphi_-(\mathbf{k})+\varphi_-(\mathbf{k}_*)=h(\mathbf{k}+\mathbf{k}_*).
\end{equation*}
If we choose $\mathbf{k}_*=\mathbf{0}$ and observe that on account of  \eqref{EQ:rel2} one has $\varphi_-(\mathbf{0})=0$, we get 
\begin{equation*}
\varphi_-(\mathbf{k})=h(\mathbf{k})
\end{equation*}
wherefrom
\begin{equation}\label{EQ:rel6}
\varphi_-(\mathbf{k})+\varphi_-(\mathbf{k}_*)=\varphi_-(\mathbf{k}+\mathbf{k}_*).
\end{equation}
under the condition $\mathbf{k}\cdot\mathbf{k}_*=\vert\mathbf{k}\vert\vert\mathbf{k}_*\vert$. 

If we choose $\mathbf{k}_*=\mathbf{k}$ one has
\begin{equation*}
2\varphi_-(\mathbf{k})=\varphi_-(2\mathbf{k}).
\end{equation*}
We want to prove that
\begin{equation*}
n\varphi_-(\mathbf{k})=\varphi_-(n\mathbf{k}), \forall n\in\mathbb{N}
\end{equation*}
by induction. In fact, the formula is true  for  $n=0,1,2$. Let us suppose that  $(n-1)\varphi_-(\mathbf{k})=\varphi_-((n-1)\mathbf{k})$ holds true and evaluate
\begin{equation*}
n\varphi_-(\mathbf{k})=(n-1)\varphi_-(\mathbf{k})+\varphi_-(\mathbf{k}) =\varphi_-((n-1)\mathbf{k})+\varphi_-(\mathbf{k}).
\end{equation*}
Observing that
\begin{equation*}
(n-1)\mathbf{k}\cdot\mathbf{k}=(n-1)\vert\mathbf{k}\vert^2=\vert (n-1)\mathbf{k}\vert\vert\mathbf{k}\vert,
\end{equation*}
 by \eqref{EQ:rel6} we find
\begin{equation*}
n\varphi_-(\mathbf{k})=\varphi_-((n-1+1)\mathbf{k})=\varphi_-(n\mathbf{k}).
\end{equation*}
Let us consider now for each  $n,m\in\mathbb{N}$
\begin{equation*}
m\varphi_-\left(\frac{n}{m}\mathbf{k}\right) = \varphi_-(m\frac{n}{m}\mathbf{k})=\varphi_-(n\mathbf{k})=n\varphi_-(\mathbf{k}),
\end{equation*}
which implies
\begin{equation*}
\varphi_-(q\mathbf{k})=q\varphi_-(\mathbf{k}), \qquad \forall q\in\mathbb{Q}^+
\end{equation*}
The functions $\varphi_-(q\mathbf{k})$ and $q\varphi_-(\mathbf{k})$ are continuous and have the same values for each  $q\in\mathbb{Q}$. By a density argument,  $\varphi_-(\alpha\mathbf{k})=\alpha\varphi_-(\mathbf{k})$ for each $\alpha\in\mathbb{R}$, that is $\varphi_-$ is a homogeneous function in $\R$.

By virtue of  Lemma \ref{LEMMA:2} and the differentiability of  $\varphi_-$,  it follows that  $\varphi_-$ is linear and therefore by Lemma \ref{LEMMA:1} we have 
\begin{equation}
\varphi_-(\mathbf{k})=2\mathbf{b}\cdot\mathbf{k},
\end{equation}
for a suitable constant vector  $\mathbf{b}$. Since $\varphi=(\varphi_+ + \varphi_-)/2$, the proof is complete.
\end{proof}

\begin{remark}
Going back to the frame centered in $\Gamma$, the collisional invariants take the form
\begin{equation}
\varphi(\mathbf{k}) = a + \mathbf{b}\cdot(\mathbf{k}-\mathbf{k}_l)+c\vert\mathbf{k}-\mathbf{k}_l\vert
\end{equation}
or equivalently
\begin{equation}
\varphi(\mathbf{k}) = a + \mathbf{b}\cdot\mathbf{k}+c\vert\mathbf{k}-\mathbf{k}_l\vert.
\end{equation}\end{remark}

\begin{remark}
The expression (\ref{coll_inv}) is formally the same as that found in the relativistic case (see \cite{Cercignani} equation (2.56) for particles having zero mass). However, some crucial differences are present.  The proof we propose requires only the differentiability of the odd part of the collisional invariant while the proof in \cite{Cercignani} is based on the assumption that $\varphi$ is $C^2$ which is not true in the case under investigation.

Our results can be cast in the form expressed by equation (2.56) in  \cite{Cercignani} by replacing the velocity of light $c$ with the Fermi velocity.  It is clear that using Lorentz transformations with $v_F$ instead of $c$ has not a physical rationale because it should mean that the electron velocity inside graphene is invariant with respect to inertial observers.  To our knowledge there is no experimental evidence of that. 
\end{remark}

Before studying the kernel of the electron-electron collision operator, we recall the following symmetry properties  that can be directly verified from the expression of $S_{ee}$:
\begin{eqnarray}
S_{ee}(\mathbf{k}_1,\mathbf{k}_2,\mathbf{k}'_1,\mathbf{k}'_2) = S_{ee}(\mathbf{k}_2,\mathbf{k}_2,\mathbf{k}'_1,\mathbf{k}'_2) =
S_{ee}(\mathbf{k}'_1,\mathbf{k}'_2,\mathbf{k}_1,\mathbf{k}_2). \label{See_symmetries}
\end{eqnarray}
\begin{corollary}
The kernel of the electron-electron collision operator is given by the function of the form
\begin{equation}
f(\mathbf{k}) = \frac{1}{1 + \exp \left(a + \mathbf{b}\cdot\mathbf{k}+c\vert\mathbf{k}-\mathbf{k}_l\vert  \right)}  \label{equilibriumg}
\end{equation}
with $a, c \in \R$, $\mathbf{b} \in \R^2$.
\end{corollary}
\begin{proof}
Let us denote by $Q_{ee}$  the electron-electron collision operator for a single band.
We have
\begin{eqnarray*}
Q_{ee} & = &\sum_{\mathbf{k}_*, \mathbf{k}', \mathbf{k}'_*} \left[ S_{ee} (\mathbf{k}',\mathbf{k}'_*,\mathbf{k},\mathbf{k}_*) f (\mathbf{k}') f(\mathbf{k}'_*)  (1-f(\mathbf{k}))(1-f(\mathbf{k}_*))\right. \nonumber\\
& & \left. - S_{ee}(\mathbf{k},\mathbf{k}_*,\mathbf{k}',\mathbf{k}'_*) f(\mathbf{k})f(\mathbf{k}_*) (1-f(\mathbf{k}'))(1-f(\mathbf{k}'_*)) \right]
\, \delta_{\mathbf{k} + \mathbf{k}_* = \mathbf{k}' + \mathbf{k}'_*}\\
& & =  \int  \sum_{\mathbf{k}_*, \mathbf{k}'} \left[ S_{ee} (\mathbf{k}',\mathbf{k}'_*,\mathbf{k},\mathbf{k}_*) f (\mathbf{k}') f(\mathbf{k}'_*)  (1-f(\mathbf{k}))(1-f(\mathbf{k}_*))\right. \nonumber\\
& &\left. - S_{ee}(\mathbf{k},\mathbf{k}_*,\mathbf{k}',\mathbf{k}'_*) f(\mathbf{k})f(\mathbf{k}_*) (1-f(\mathbf{k}'))(1-f(\mathbf{k}'_*)) \right] 
\delta \left(\mathbf{k} + \mathbf{k}_* -  \mathbf{k}' - \mathbf{k}'_*  \right).
d \mathbf{k}'_*
\end{eqnarray*}
After homogenization \cite{Jacoboni} $Q_{ee}$ becomes
\begin{eqnarray*}
& &\frac{1}{(2 \pi)^4}
\int  \left[ \sigma_{ee} (\mathbf{k}',\mathbf{k}'_*,\mathbf{k},\mathbf{k}_*) f (\mathbf{k}') f(\mathbf{k}'_*)  (1-f(\mathbf{k}))(1-f(\mathbf{k}_*))\right. \nonumber\\
& &\left. - \sigma_{ee}(\mathbf{k},\mathbf{k}_*,\mathbf{k}',\mathbf{k}'_*) f(\mathbf{k})f(\mathbf{k}_*) (1-f(\mathbf{k}'))(1-f(\mathbf{k}'_*)) \right] \\
& &
\times \delta \left(\varepsilon(\mathbf{k})  + \varepsilon(\mathbf{k}_*)  - \varepsilon(\mathbf{k}')  - \varepsilon(\mathbf{k}'_*) \right)
 d \mathbf{k}_*  d \mathbf{k}'  d \mathbf{k}'_*,
\end{eqnarray*}
where 
$$
 \sigma_{ee}(\mathbf{k},\mathbf{k}_*,\mathbf{k}',\mathbf{k}'_*) = \frac{2\pi}{\hbar}\vert M  (\mathbf{k}',\mathbf{k}'_*,\mathbf{k},\mathbf{k}_*) \vert^2 \delta \left(\mathbf{k} + \mathbf{k}_* -  \mathbf{k}' - \mathbf{k}'_*  \right) 
$$
which has the same symmetry properties of $S_{ee}$.

By standard arguments (see for example \cite{AbDe,Juengel}) one can prove that for any function $g(\mathbf{k})$ regular enough the following relation holds
\begin{eqnarray}
&&\int_{\R^2} Q_{ee} (f) g(\mathbf{k}) d \mathbf{k}  = \frac{1}{(2 \pi)^4}
\int \frac{\sigma_{ee}(\mathbf{k}, \mathbf{k}_*, \mathbf{k}',\mathbf{k}'_*)}{4} \delta \left(\varepsilon(\mathbf{k})  + \varepsilon(\mathbf{k}_*)  - \varepsilon(\mathbf{k}')  - \varepsilon(\mathbf{k}'_*) \right) \nonumber\\
&&\times \left(f(\mathbf{k}') f (\mathbf{k}'_*) (1 - f (\mathbf{k}))(1 - f(\mathbf{k}_*))   - f (\mathbf{k}) f (\mathbf{k}_*) (1 - f (\mathbf{k}'))(1 - f (\mathbf{k}'_*))  \right) \nonumber\\
&& \times \left(g(\mathbf{k}) + g(\mathbf{k}_*) - g(\mathbf{k}') - g(\mathbf{k}'_*) \right)  d \mathbf{k}  d \mathbf{k}_*  d \mathbf{k}'  d \mathbf{k}'_*. \label{int_collision}
\end{eqnarray}
With the choice 
$$
g(\mathbf{k}) = H (f(\mathbf{k})) = \log \frac{f(\mathbf{k})}{1 - f(\mathbf{k})}
$$
it is easy to verify that 
$$
\int_{\R^2} Q_{ee} (f)  \log \frac{f(\mathbf{k})}{1 - f(\mathbf{k})} d \mathbf{k} \le 0
$$
and 
$$
\int_{\R^2} Q_{ee} (f)  \log \frac{f(\mathbf{k})}{1 - f(\mathbf{k})} d \mathbf{k} = 0
$$
if and only if 
$$
H (f(\mathbf{k})) + H (f(\mathbf{k}')) = H (f(\mathbf{k}_1)) + H (f(\mathbf{k}_1'))
$$
provided the relations (\ref{conservation1})-(\ref{conservation2}) are satisfied. Therefore, $H (f(\mathbf{k}))$ must be a collisional invariant. On account of the previous theorem
$$
H (f(\mathbf{k})) = a + \mathbf{b}\cdot\mathbf{k}+c\vert\mathbf{k}-\mathbf{k}_l\vert,
$$
with $a, c \in \R$, $\mathbf{b} \in \R^2$  arbitrary,
wherefrom the expression (\ref{equilibriumg}) is obtained. 
\end{proof}
\begin{remark}
If we set $a = - \dfrac{\varepsilon_F}{k_B T} $, $\mathbf{b} = \mathbf{0}$ and $c = \pm \dfrac{\hbar v_F}{k_B T}$,  where $T$ is the crystal temperature, $\varepsilon_F$ is the Fermi energy and $k_B$ is the Boltzmann constant, the standard Fermi-Dirac distribution 
\begin{equation}
f(\mathbf{k}) = \frac{1}{1 + \exp \left(\dfrac{\varepsilon(\vert\mathbf{k}-\mathbf{k}_l\vert )-  \varepsilon_F}{k_B T}  \right)}, 
\end{equation}
is recovered, the  case of positive and negative sign holding in the conduction and valence band respectively.
\end{remark}
\begin{remark}
By setting $g(\mathbf{k}) = 1$ from (\ref{int_collision}) one has the charge conservation (in the unipolar case)
$$\int_{\R^2} Q_{ee} (f)  d \mathbf{k}  = 0,$$
by setting $g(\mathbf{k}) = \mathbf{k}$ from (\ref{int_collision}) one has the  conservation of the crystal momentum$$\int_{\R^2} \mathbf{k}\,  Q_{ee} (f)  d \mathbf{k}  = {\bf 0},$$
while by setting $g(\mathbf{k}) =  \varepsilon(\mathbf{k})$ from (\ref{int_collision}) one has the energy conservation
$$\int_{\R^2} \varepsilon(\mathbf{k}) \, Q_{ee} (f) d \mathbf{k}  = 0.$$
\end{remark}
\begin{remark}
For the electron-phonon scattering by following \cite{Majorana,AbDe}, it is possible to see that the collisional invariants are functions of the type
$$
\varphi(\mathbf{k}) = a + c \vert\mathbf{k}-\mathbf{k}_l\vert.
$$ 
and the kernel is spanned by the functions of the form
\begin{equation}
f(\mathbf{k}) = \frac{1}{1 + \exp \left(a + c\vert\mathbf{k}-\mathbf{k}_l\vert  \right)}.
\end{equation}
\end{remark}

\section{Electron-electron scattering rate: the intra-band case}

Now we want to address the importance of the electron-electron scattering by Monte Carlo simulations. Other approaches can be used but the complexity
of the electron electron collision term makes direct simulations, e.g. based on WENO scheme \cite{LMS} or Discontinuous Galerkin methods \cite{NaCaRo},  a daunting computational task, unless drastic simplifications are introduced \cite{Battiato1,Battiato2}. We will adopt the Monte Carlo scheme presented in \cite{RoMajCo} which allows a correct inclusion of the Pauli exclusion principle.
The details are omitted, the interested reader is referred to reference   \cite{RoMajCo,CoMajRo,CoRo}.
Both the intra-band and the inter-band cases will be tackled. 

First we consider intra-band scattering.

The scattering rate due to the electron-electron interaction  is given by
\begin{align*}
\lambda_{ee} (\mathbf{k}_1) & = \sum_{\mathbf{k}_2} f(\mathbf{k}_2) \sum_{\mathbf{k}'_1,\mathbf{k}'_2} S_{ee}(\mathbf{k}_1,\mathbf{k}_2,\mathbf{k}'_1,\mathbf{k}'_2).
\end{align*}
under the conditions of momentum and energy conservation 
\begin{equation}\label{EQ:cons_mom}
\mathbf{k}_1+\mathbf{k}_2 = \mathbf{k}'_1+\mathbf{k}'_2,
\end{equation}
\begin{equation}\label{EQ:cons_en}
\vert\mathbf{k}_1\vert+\vert\mathbf{k}_2\vert = \vert\mathbf{k}'_1\vert+\vert\mathbf{k}'_2\vert,
\end{equation}
with $\mathbf{k}_1$, $\mathbf{k}_2$ states before the scattering and $\mathbf{k}'_1$, $\mathbf{k}'_2$ states after the collision. 

The previous relations imply that, given the initial states, the final ones must belong to an ellipse $\Gamma$ (see Fig. \ref{Schematic_gamma}). Note that this is true both in the conduction and in the valence band.

If we denote by $a$, $b$ e $c$ the length of the major, minor  and focal semi-axis, assigned  $\mathbf{k}_1$ and $\mathbf{k}_2$, one has
\begin{equation*}
\vert \mathbf{k}_1 \vert + \vert \mathbf{k}_2 \vert = 2a.
\end{equation*}
Moreover,  the parallelogram rule gives
\begin{equation*}
\vert \mathbf{k}_1 + \mathbf{k}_2 \vert = 2c
\end{equation*}
 and from the relation $a^2=b^2+c^2$ one gets $b$. 

Under the above considerations,
\begin{align*}
\lambda_{ee} (\mathbf{k}_1) & = \sum_{\mathbf{k}_2} f(\mathbf{k}_2) \sum_{\mathbf{k}'_1,\mathbf{k}'_2} \frac{2\pi}{\hbar} \vert M \vert^2 \delta(\varepsilon(\mathbf{k}'_1)+\varepsilon(\mathbf{k}'_2)-\varepsilon(\mathbf{k}_1)-\varepsilon(\mathbf{k}_2)) =\\
& = \sum_{\mathbf{k}_2} f(\mathbf{k}_2) \sum_{\mathbf{k}'_1} \frac{2\pi}{\hbar} \vert M \vert^2(q,q') \delta(\varepsilon(\mathbf{k}'_1)+\varepsilon(\mathbf{k}_1+\mathbf{k}_2-\mathbf{k}'_1)-\varepsilon(\mathbf{k}_1)-\varepsilon(\mathbf{k}_2)).
\end{align*}
and by homogenization \cite{Jacoboni}  we have
\begin{align*}
\lambda_{ee} (\mathbf{k}_1) & \approx \frac{1}{(2\pi)^4} \int_{\mathbb{R}^2}\left[  f(\mathbf{k}_2) \int_{\mathbb{R}^2} \! \frac{2\pi}{\hbar^2 v_F} \vert M \vert^2(q,q') \delta(\vert\mathbf{k}'_1\vert+\vert\mathbf{k}_1+\mathbf{k}_2-\mathbf{k}'_1\vert-\vert\mathbf{k}_1\vert-\vert\mathbf{k}_2\vert) \, \diff \mathbf{k}'_1 \right] \, \diff \mathbf{k}_2.
\end{align*}
 In the wave-vector $\hbar\mathbf{k}$ space let us consider a mesh made of 
 $N\times N$ square cells of edge $\Delta_{\hbar\mathbf{k}}$. 
 If we characterize each cell by a couple of integers $(i,j)$ and denote by $f_{ij}$ the mean value of $f(\mathbf{k}_2)$ in the cell, one can use the middle point approximation 
\begin{align*}
\lambda_{ee} (\mathbf{k}_1) & \approx \frac{1}{(2\pi)^3\hbar^2 v_F} \sum_{i,j=1}^{N} f_{ij} \frac{\Delta_{\hbar\mathbf{k}}^2}{\hbar^2} \left[   \int_{\mathbb{R}^2} \! \vert M \vert^2(q,q') \delta(\vert\mathbf{k}'_1\vert+\vert\mathbf{k}_1+\mathbf{k}_2-\mathbf{k}'_1\vert-\vert\mathbf{k}_1\vert-\vert\mathbf{k}_2\vert) \, \diff \mathbf{k}'_1 \right] =\\
& = \frac{\Delta_{\hbar\mathbf{k}}^2}{(2\pi)^3\hbar^4 v_F} \sum_{i,j=1}^{N} f_{ij} \int_\Gamma \! \vert M \vert^2(q,q') \, \diff \Gamma.
\end{align*}
To complete the evaluation we need a parametric representation of  $\Gamma$.

Let us consider an electron having wave-vector $\mathbf{k}_1$ and associated with it a companion electron having wave-vector $\mathbf{k}_2$. Besides the initial reference frame
$\vec{k}_x O \vec{k}_y$, 
let us consider the reference frame $\vec{k}'_x F_1 \vec{k}'_y$, where $F_1$ is one of the two foci,  $\vec{k}'_x \parallel \mathbf{k}_1+\mathbf{k}_2$ and $\vec{k}'_y \perp \vec{k}'_x$.
Without loss of generality we can choose $ O \equiv F_1 \equiv K$ with $K$ Dirac point (see Fig. \ref{Schematic_gamma}). Let us denote by $\theta$ the angle between $\vec{k}'_x$ and $\vec{k}_x$. 

\begin{figure}[ht]
\centering
\includegraphics[height=8.0cm]{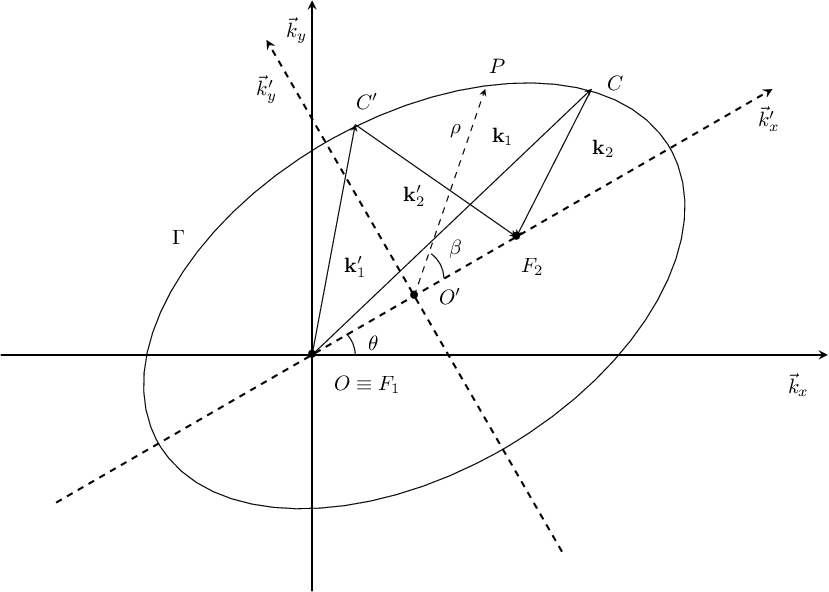} 
\caption{Schematic representation of the electron states before and after the scattering in the original and elliptic reference frames.} \label{Schematic_gamma}
\end{figure}

Given a point $P$ belonging to the ellipses, let us denote by $\beta$ the angle between   $P-O'$ the $\vec{k}'_x$-axes (see Fig. \ref{Schematic_gamma}), $O'$ being the center of the ellipses. 

The parametric equations of the $\Gamma$ are  
\begin{equation*}
\Gamma:\left\lbrace
\begin{aligned}
\vec{k}_x & = \frac{(\mathbf{k}_1+\mathbf{k}_2)_x}{2} + a\cos\beta\cos\theta-b\sin\beta\sin\theta\\
\vec{k}_y & = \frac{(\mathbf{k}_1+\mathbf{k}_2)_y}{2} + a\cos\beta\sin\theta+b\sin\beta\cos\theta
\end{aligned}
\right.  \qquad \beta \in [0,2\pi[
\end{equation*}
and
\begin{align*}
\diff\Gamma  = \sqrt{\left(\frac{\diff \vec{k}_x}{\diff\beta}\right)^2+\left(\frac{\diff \vec{k}_y}{\diff\beta}\right)^2} \, d \beta
 = \sqrt{a^2\sin^2\beta+b^2\cos^2\beta}  \, d \beta.
\end{align*}
The post collision vectors $\mathbf{k}'_1$ and $\mathbf{k}'_2$ are given by
\begin{equation*}
\mathbf{k}'_1 = C'-F_1 = \left(\frac{(\mathbf{k}_1+\mathbf{k}_2)_x}{2} + a\cos\beta\cos\theta-b\sin\beta\sin\theta , \frac{(\mathbf{k}_1+\mathbf{k}_2)_y}{2} + a\cos\beta\sin\theta+b\sin\beta\cos\theta\right),
\end{equation*}
\begin{equation*}
\mathbf{k}'_2 = F_2-C' = \left(\frac{(\mathbf{k}_1+\mathbf{k}_2)_x}{2} - a\cos\beta\cos\theta+b\sin\beta\sin\theta , \frac{(\mathbf{k}_1+\mathbf{k}_2)_y}{2} - a\cos\beta\sin\theta-b\sin\beta\cos\theta\right).
\end{equation*}
Therefore, the electron-electron scattering rate can be evaluated as
\begin{align*}
\lambda_{ee} & \approx \frac{\Delta_{\hbar\mathbf{k}}^2}{(2\pi)^3\hbar^4 v_F} \sum_{i,j=1}^{N} f_{ij} \int_\Gamma \! \vert M \vert^2(q,q') \, \diff \Gamma =\\
& = \frac{\Delta_{\hbar\mathbf{k}}^2}{(2\pi)^3\hbar^4 v_F} \sum_{i,j=1}^{N} f_{ij} \int_0^{2\pi} \! \vert M \vert^2(q(\beta),q'(\beta)) \sqrt{a^2\sin^2\beta + b^2\cos^2\beta} \, \diff \beta,
\end{align*}
where 
\begin{equation*}
q(\beta)=\vert \mathbf{k}_1-\mathbf{k}'_1 \vert, \quad 
q'(\beta) = \vert \mathbf{k}_1-\mathbf{k}'_2\vert.
\end{equation*}

Introducing an uniform mesh of $[0,2\pi]$ of size $\Delta \beta = 2\pi / m$, $0 = \beta_0 <\beta_1 < \ldots < \beta_k < \beta_{k+1} < \ldots < \beta_{m} = 2\pi$,
and applying the trapezoidal rule, one gets  
\begin{align*}
\lambda_{ee} & \approx \frac{\Delta_{\hbar\mathbf{k}}^2}{(2\pi)^3\hbar^4 v_F} \sum_{i,j=1}^{N} f_{ij} \frac{\Delta \beta}{2} \sum_{k=1}^m \left[ g(\beta_{k-1})+g(\beta_k) \right]
%
& = C_{ee} \sum_{i,j=1}^{N} f_{ij} \sum_{k=1}^m \left[ \widetilde{g}(\beta_{k-1})+\widetilde{g}(\beta_k) \right],
\end{align*}
where
\begin{eqnarray*}
& &C_{ee} = \frac{r_s^2 \kappa^2 v_F g_s g_v \Delta_{\hbar\mathbf{k}}^2 \Delta \beta}{512\pi\hbar^2}\\
& &\widetilde{g}(\beta_k) = \vert \widetilde{M} \vert^2(q(\beta_k),q'(\beta_k)) \sqrt{a^2\sin^2\beta_k + b^2\cos^2\beta_k},\\
& & \vert \widetilde{M} \vert^2 (q(\beta_k),q'(\beta_k)) 
= \vert\widetilde{V}(q(\beta_k))\vert^2 + \vert\widetilde{V}(q'(\beta_k))\vert^2 -\widetilde{V}(q(\beta_k))\widetilde{V}(q'(\beta_k)),\\
& &\widetilde{V}(q(\beta_k)) = \frac{1}{\epsilon(q(\beta_k))q(\beta_k)}\left( 1+\cos(\phi_{\mathbf{k}_1,\mathbf{k}'_1}) \right)\left( 1+\cos(\phi_{\mathbf{k}_2,\mathbf{k}'_2}) \right),\\
& &\widetilde{V}(q'(\beta_k)) = \frac{1}{\epsilon(q'(\beta_k))q'(\beta_k)}\left( 1+\cos(\phi_{\mathbf{k}_1,\mathbf{k}'_2}) \right)\left( 1+\cos(\phi_{\mathbf{k}_2,\mathbf{k}'_1}) \right),\\
& &\epsilon(q(\beta_k))q(\beta_k)  = 
 q(\beta_k)+C_\epsilon\widetilde{\Pi}(q(\beta_k)),
\end{eqnarray*}
with
\begin{eqnarray}
& &C_\epsilon = \frac{r_s k_F}{2}(g_s g_v)^{3/2},
\end{eqnarray}
\begin{equation*}
\widetilde{\Pi}(q(\beta_k)) = \left\lbrace
\begin{alignedat}{2}
&1 && \qquad\mbox{if}\quad q(\beta_k)<2k_F\\
&1+\frac{\pi q(\beta_k)}{8 k_F}-\frac{\sqrt{q^2(\beta_k)-4k_F^2}}{2q(\beta_k)}-\frac{q(\beta_k)}{4k_F}\arcsin\left( \frac{2k_F}{q(\beta_k)} \right) && \qquad\mbox{otherwise}
\end{alignedat}
\right.
\end{equation*}
$k_F=\sqrt{4\pi n/g_v g_s}$, $n$ being the  electron density.

Once a numerical approximation of the transition rate has been obtained, the Monte Carlo simulation proceeds in a  standard way. If the electron-electron scattering is selected after a free flight, to the wave-vector 
 $\mathbf{k}_1$ of the considered electron, a companion electron is associated according to the discrete estimation $f_{ij}$ of $f$ at the current time. 
 This is equivalent selecting with a uniform distribution one of the simulated particles, of course apart from the first one. 
Let $\mathbf{k}_2$ be the wave-vector of the companion electron. With the procedure delineated above one determines the ellipses and randomly chooses the angle
$\beta \in [0,2\pi]$ selecting the point $P$ and, therefore, the wave-vector of the outgoing electrons $\mathbf{k}'_1$ and $\mathbf{k}'_2$. Eventually, in order to take into account the Pauli exclusion principle, one generates two random numbers $\eta_1, \eta_2  \in [0,1]$. If 
 $f(\mathbf{k}'_1)<\eta_1$ and $f(\mathbf{k}'_2)<\eta_2$ the transition is accepted  and $\mathbf{k}_1$ and  $\mathbf{k}_2$ become $\mathbf{k}'_1$ and $\mathbf{k}'_2$, otherwise the initial states remain unchanged.

\section{Bipolar model}

As shown in \cite{ART:Brida}, the only allowed interband electron-electron scattering is when the two electrons are initially in different bands  and the final states remain in different bands. In particular, carrier multiplication, that is a transition which creates an excess of charge in the conduction band, and Auger recombination, which decreases the charge in the conduction band, has zero measure in the phase space.

To fix the ideas, without loss of generality let us suppose that $\mathbf{k}_1$ is in the conduction band (CB) while  $\mathbf{k}_2$ is in the valence band (VB) and  $\mathbf{k}'_1$ is in CB and  $\mathbf{k}'_2$ is in VB. 

Momentum and energy conservations now lead to the relations
\begin{equation*}
\mathbf{k}_1+\mathbf{k}_2 = \mathbf{k}'_1+\mathbf{k}'_2,
\end{equation*}
\begin{equation*}
\vert \mathbf{k}_1 \vert - \vert \mathbf{k}_2 \vert = \vert \mathbf{k}'_1 \vert - \vert \mathbf{k}'_2 \vert.
\end{equation*}
which determine a hyperbola $\tilde{\Gamma}$ with distance between foci $2 c$ equal to $\vert  \mathbf{k}_1+\mathbf{k}_2 \vert$.

\begin{figure}[ht]
\centering
\includegraphics[height=8.0cm]{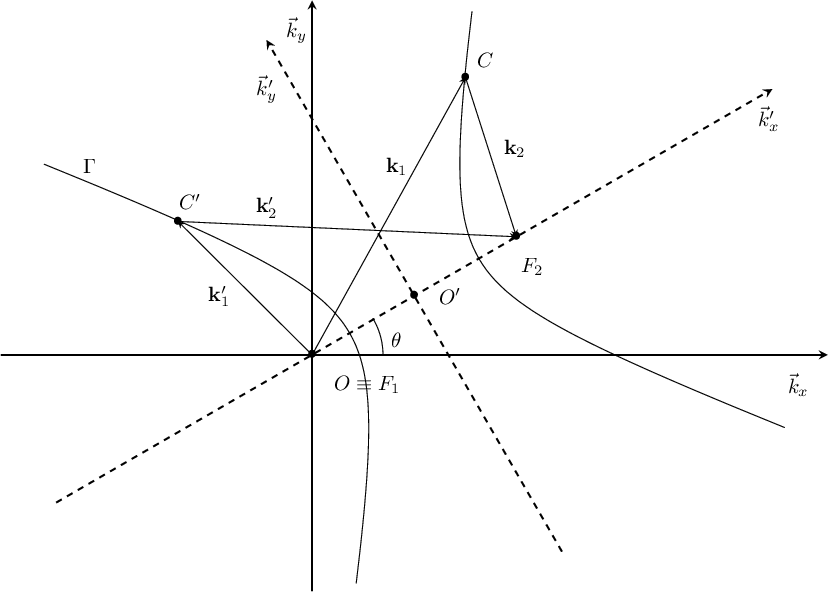} 
\caption{Schematic representation of the hyperbola arising in an electron-electron scattering rate in the interband case.}
\label{hyperbola}
\end{figure}

If we denote by  $a$ the length of the major semi-axis, one has
\begin{equation*}
\vert \vert \mathbf{k}_1 \vert - \vert \mathbf{k}_2 \vert \vert= 2a.
\end{equation*}

and from the relation  $c^2=a^2+b^2$  one gets the length $b$ of the minor semi-axis.

The electron-electron scattering rate in the interband case reads 
\begin{align*}
\lambda_{ee} 
& = \sum_{\mathbf{k}_2} f(\mathbf{k}_2) \sum_{\mathbf{k}'_1} \frac{2\pi}{\hbar} \vert M \vert^2(q,q') \delta(\varepsilon(\mathbf{k}'_1)+\varepsilon(\mathbf{k}_1+\mathbf{k}_2-\mathbf{k}'_1)-\varepsilon(\mathbf{k}_1)-\varepsilon(\mathbf{k}_2))
\end{align*}
and by homogenization \cite{Jacoboni}, one has
\begin{align*}
\lambda_{ee} & \approx 
\frac{1}{(2\pi)^3}\int_{\mathbb{R}^2}\left[  f(\mathbf{k}_2) \int_{\mathbb{R}^2} \! \frac{1}{\hbar^2 v_F} \vert M \vert^2(q,q') \delta(\vert\mathbf{k}'_1\vert-\vert\mathbf{k}_1+\mathbf{k}_2-\mathbf{k}'_1\vert-\vert\mathbf{k}_1\vert+\vert\mathbf{k}_2\vert) \, \diff \mathbf{k}'_1 \right] \, \diff \mathbf{k}_2.
\end{align*}
If we discretize the wave-vector space as in the intraband case and denote by
$f_{ij}$ the average of  $f(\mathbf{k}_2)$ in the cell determined by the indices $(i,j)$, one gets
\begin{align*}
\lambda_{ee} & \approx \frac{1}{(2\pi)^3\hbar^2 v_F} \sum_{i,j=1}^{N} f_{ij} \frac{\Delta_{\hbar\mathbf{k}}^2}{\hbar^2} \left[   \int_{\mathbb{R}^2} \! \vert M \vert^2(q,q') \delta(\vert\mathbf{k}'_1\vert-\vert\mathbf{k}_1+\mathbf{k}_2-\mathbf{k}'_1\vert-\vert\mathbf{k}_1\vert+\vert\mathbf{k}_2\vert) \, \diff \mathbf{k}'_1 \right] =\\
& = \frac{\Delta_{\hbar\mathbf{k}}^2}{(2\pi)^3\hbar^4 v_F} \sum_{i,j=1}^{N} f_{ij} \int_{\hat{\Gamma}} \! \vert M \vert^2(q,q') \, \diff \hat{\Gamma}.
\end{align*}
To complete the evaluation we need a parametric representation of $\hat{\Gamma}$.

%
Without loss of generality, we can assume that the foci of the hyperbola  have coordinates
\begin{equation*}
F_1=(0,0), \qquad F_2=\left( (\mathbf{k}_1+\mathbf{k}_2)_x, (\mathbf{k}_1+\mathbf{k}_2)_y \right).
\end{equation*}
The center is given by
\begin{equation*}
O'=\left( \frac{(\mathbf{k}_1+\mathbf{k}_2)_x}{2}, \frac{(\mathbf{k}_1+\mathbf{k}_2)_y}{2} \right).
\end{equation*}
A new reference frame, said \emph{hyperbolic}, appears (see Fig. \ref{hyperbola}). It is centered in $O'$ and has axes such that  $\vec{k}'_x \parallel \mathbf{k}_1+\mathbf{k}_2$ e $\vec{k}'_y \perp \vec{k}'_x$. We denote by $\theta$ the angle between the axes $\vec{k}'_x$ and $\vec{k}_x$. 

In the hyperbolic reference frame,  a parametric representation of $\hat{\Gamma}$ reads
\begin{equation*}
\Gamma_1:\left\lbrace
\begin{aligned}
\vec{k}'_x & = a\cosh\beta\\
\vec{k}'_y & = b\sinh\beta
\end{aligned}
\right. , \quad \beta \in \mathbb{R}; \qquad \Gamma_2: \left\lbrace
\begin{aligned}
\vec{k}'_x & = -a\cosh\beta\\
\vec{k}'_y & = b\sinh\beta
\end{aligned}
\right. , \quad \beta \in \mathbb{R},
\end{equation*}
where $\Gamma_1$ and $\Gamma_2$ represent the two branches of the hyperbola. Therefore
\begin{equation}
\lambda_{ee} = \frac{\Delta_{\hbar\mathbf{k}}^2}{(2\pi)^3\hbar^4 v_F} \sum_{i,j=1}^{N} f_{ij} \left[ \int_{\Gamma_1} \! \vert M \vert^2(q,q') \, \diff \Gamma_1 +\int_{\Gamma_2} \! \vert M \vert^2(q,q') \, \diff \Gamma_2\right] . \label{rate_ee_inter}
\end{equation}
with
\begin{align*}
\diff\Gamma_i   = \sqrt{a^2\sinh^2\beta+b^2\cosh^2\beta}  \, d \beta, \quad i = 1,2.
\end{align*}
At variance with the intraband case, now the evaluation of $\lambda_{ee}$ requires an integration on a curve which is not limited. However, as observed  in \cite{NaCaRo} we expect that the electron distribution can be considered negligible outside a compact set $\Omega$ in the wave-vector space and in the simulation the integral in (\ref{rate_ee_inter}) will be performed over $\hat{\Gamma} \cap \Omega$.

\section{Simulations in suspended monolayer graphene}
As first step,  we are interested in studying how the transport properties are influenced by the electron-electron scattering in a homogeneous  suspended monolayer of graphene.   
To this aim, we look for spatially homogeneous solutions  under a constant electric field for different values of the Fermi energy. 

One can control the shift of the Fermi energy $\eps_F$  from the Dirac point  by a suitable gate voltage (see Fig. \ref{coni}), creating a kind of $n$-doping (if $\eps_F > 0$)  or $p$-doping  (if $\eps_F < 0$).  If the absolute value of the Fermi energy is more than about 0.15 eV, the interband scattering becomes negligible on account of the Pauli exclusion principle, and a unipolar simulation is enough. For lower values a full interband simulation must be performed. 

 In the following  we consider  Fermi energies such that a unipolar situation is physically justified.   
The Boltzmann equation in a single valley reads
\begin{eqnarray}
&&
\dfrac{\partial f(\te,\bk)}{\partial t} - 
\dfrac{e}{\hbar} \, \mathbf{E} \cdot \nabla_{\bk} f(\te,\bk) = C[f],
\end{eqnarray}
which remains an integral-differential equation, 2D in the wave-vector space.

 As initial condition we take the Fermi-Dirac distribution
$$
f(0, \bk) = \dfrac{1}{1 + \exp \left( \dfrac{\eps(\bk) - \eps_F}{k_{B} \, T} \right)} ,
$$
with $T$ room temperature  (300 \,K).  Fermi energy is related to the charge density, which remains constant in the unipolar case, by 
\begin{equation}
\rho(t) = \rho(0) = \dfrac{g_s g_v}{(2 \, \pi)^{2}} \int f(0, \bk) \: d \bk,
\label{density}
\end{equation}
where $g_s =2$ and $g_v=2$ are the spin and valley degeneracies, respectively. 

\begin{figure}[ht]
\begin{center}
\includegraphics[width=0.95\textwidth]{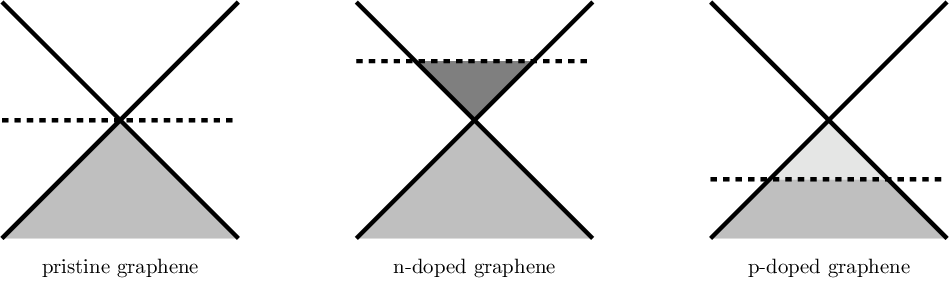} 
\caption{Shift of the Fermi energy (dashed line) with respect to the Dirac point (vertex  of the cone).}
\label{coni}
\end{center}
\end{figure}
All the electron-phonon scatterings have been included as in \cite{RoMajCo,CoMaRo,CoMajRo,CoBoDeRo,Fischetti}.
The numerical simulations with different values of the Fermi energy (and as consequence the electron density)  in a range where the unipolar description is justified.
We used the Direct Simulation Monte Carlo  devised in \cite{RoMajCo} by adding  the electron-electron scattering.
At variance with the electron-phonon scattering, for the electron-electron scattering we need to update the scattering rate after each collision event due to the presence of the electron distribution. This, of course, increases hugely the computing time. $N_p = 2\cdot 10^4$ particles have been used. As an alternative, hydrodynamical models can be formulated (the interested reader is referred to \cite{Bar,MoBa,LuMaNaRo,LuRo_IJNM,LuRo_AoP}).

In Fig.s \ref{distributions} the steady distribution functions at Fermi levels 0.15 eV (top) and 0.25 eV (bottom) are shown in the case of applied electric fields 1kV/cm (left), 3kV/cm (right). Note that the distributions never exceed one and, therefore, no violation of the Pauli exclusion principle is observed. 

\begin{figure}[ht]
\centering
\includegraphics[scale=0.5]{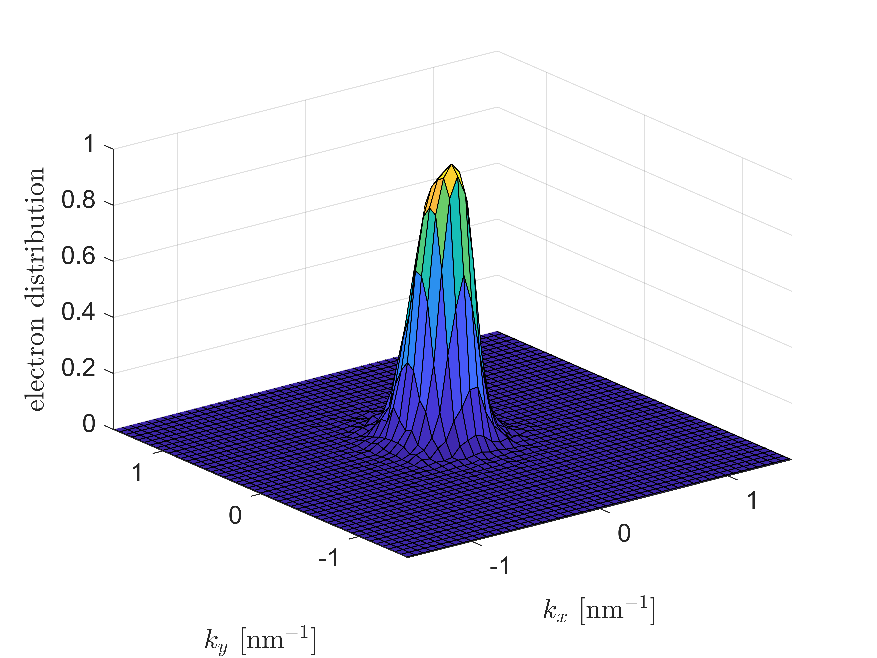}\includegraphics[scale=0.5]{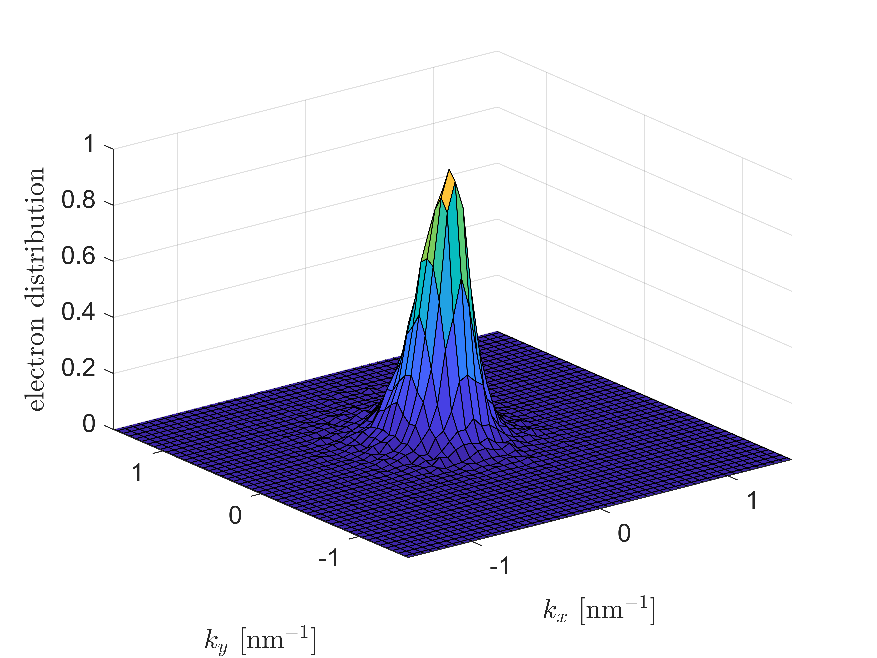}\\
\includegraphics[scale=0.5]{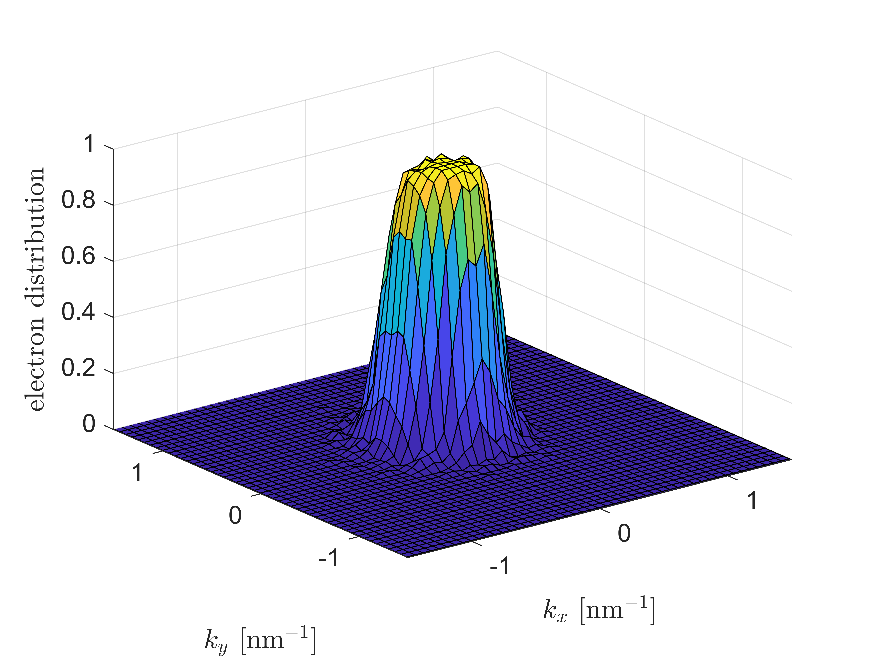}\includegraphics[scale=0.5]{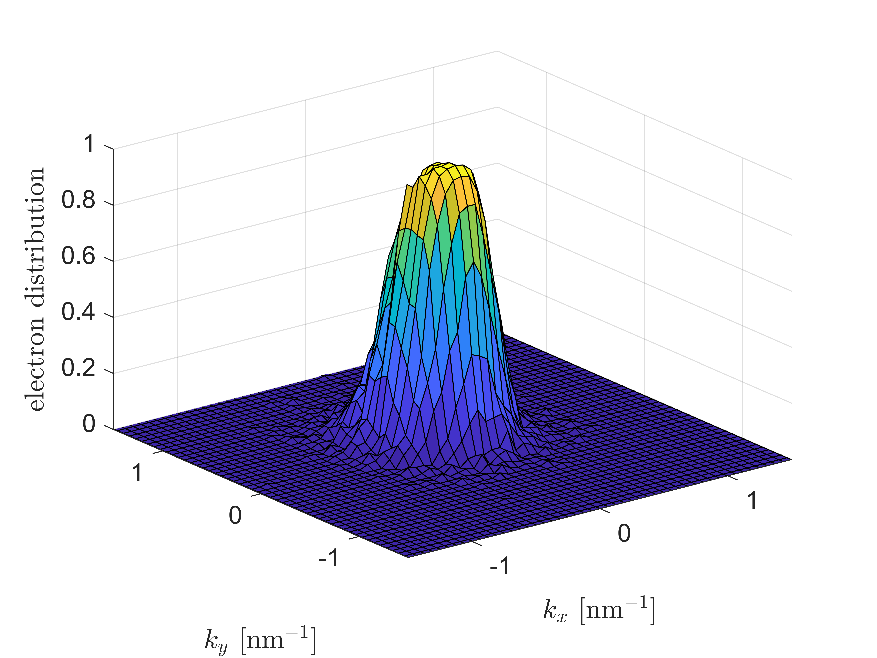}
\caption{Steady distribution functions at Fermi level 0.15 eV  (top) and 0.25 eV (bottom) with applied electric fields 1kV/cm (left), 3kV/cm (right).}
\label{distributions}
\end{figure}

In Fig.s \ref{vel_uni}, \ref{en_uni} the mean velocity $V(t)$ and mean energy $W(t)$, defined as
\begin{equation}
V(t) = \frac{1}{\rho(t) }\dfrac{g_s g_v}{(2 \, \pi)^{2}} \int {\bf v} f(t, \bk) \: d \bk, \quad
W(t) = \frac{1}{\rho(t) }\dfrac{g_s g_v}{(2 \, \pi)^{2}} \int {\bf v} f(t, \bk) \: d \bk,
\label{vel_en}
\end{equation}
are shown versus time with and without the electron-electron scattering. 
\begin{figure}[ht]
\centering
\includegraphics[scale=0.4]{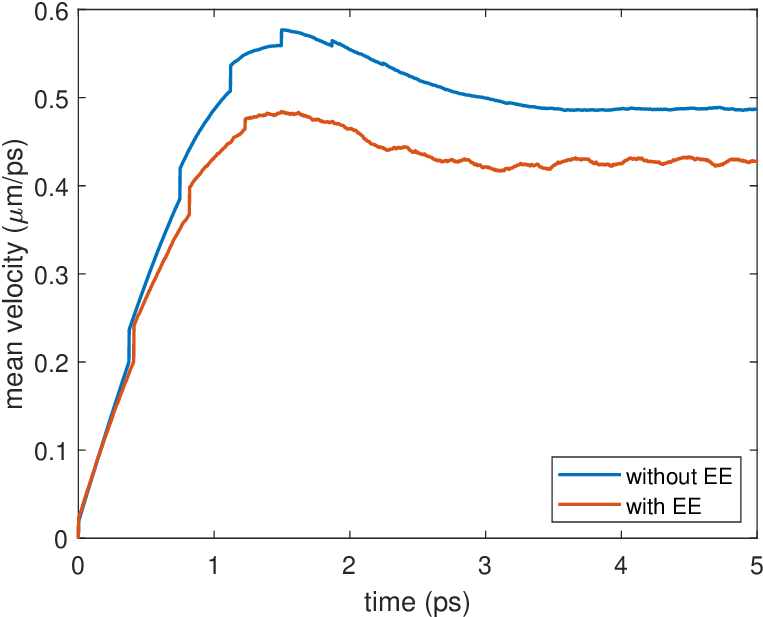}\includegraphics[scale=0.4]{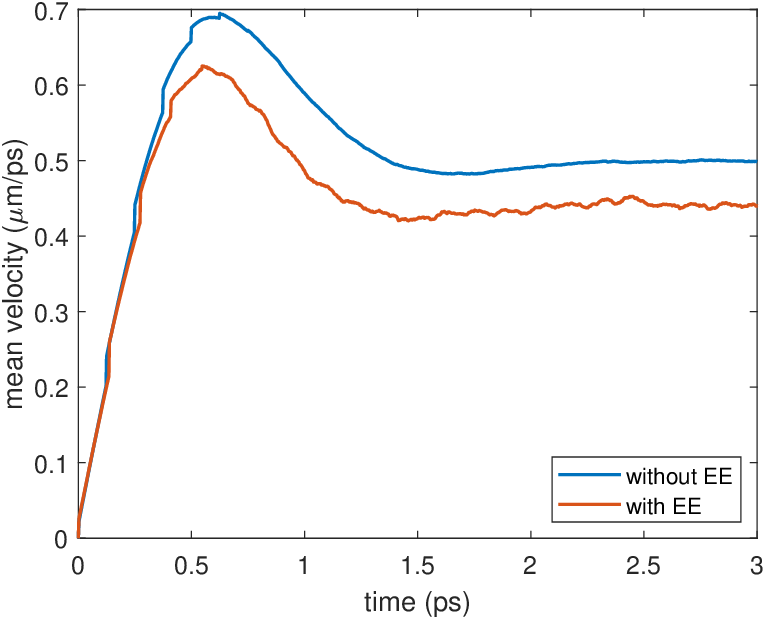}\includegraphics[scale=0.4]{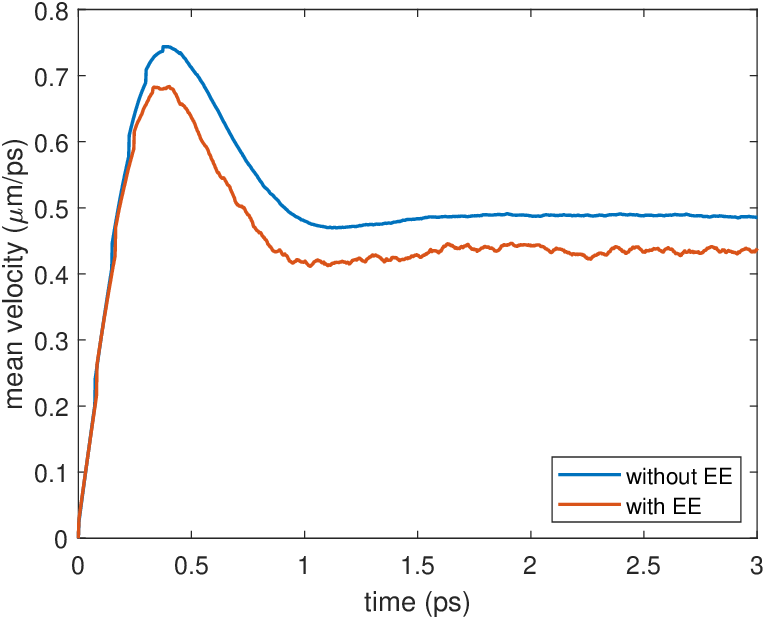}\\
\includegraphics[scale=0.4]{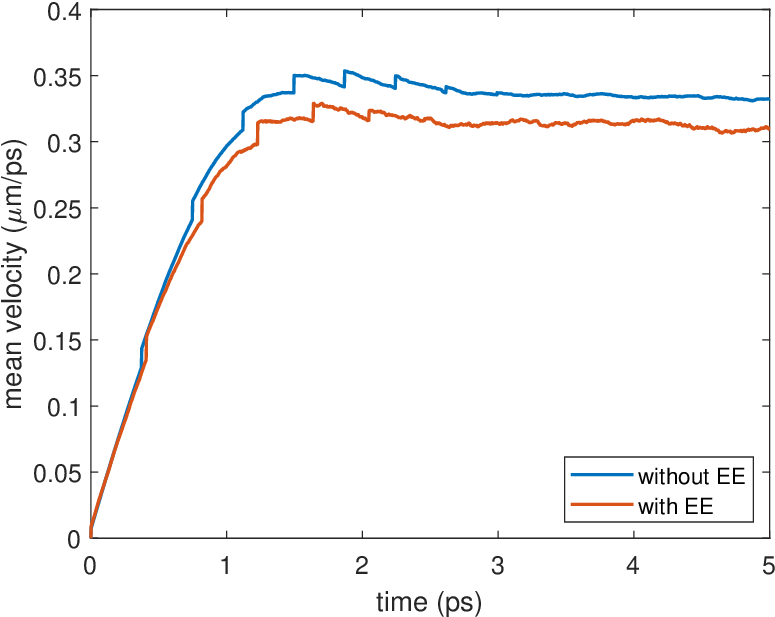}\includegraphics[scale=0.4]{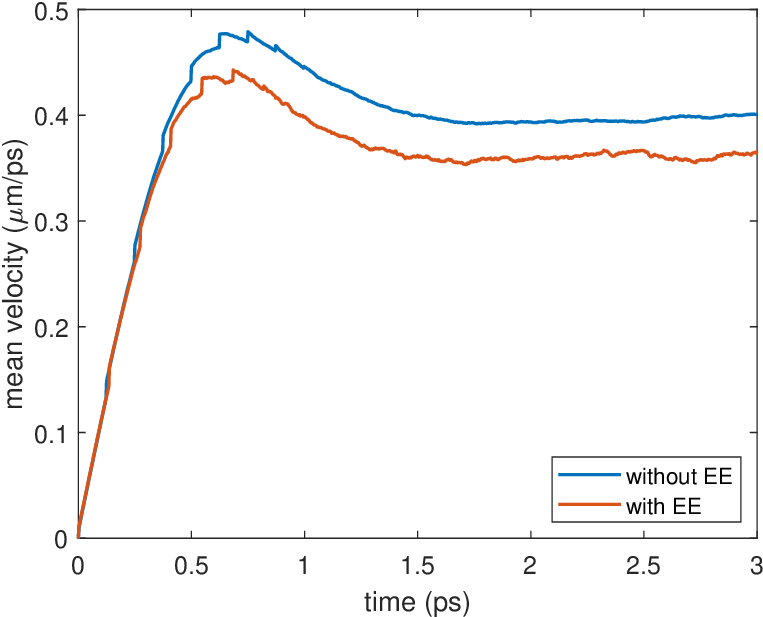}\includegraphics[scale=0.4]{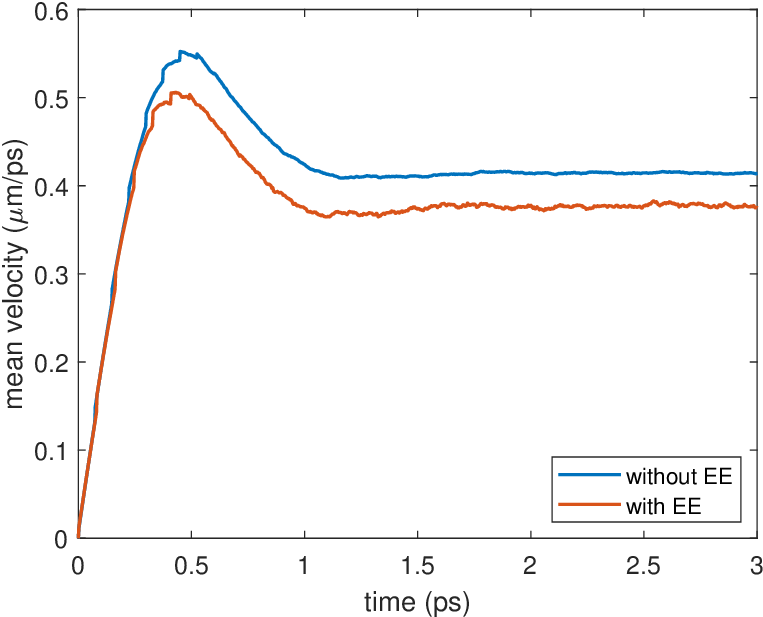}
\caption{Mean velocity versus time at Fermi level 0.15 eV  (top) and 0.25 eV (bottom) with applied electric fields 1kV/cm (left), 3kV/cm (center), 5kV/cm (right).}
\label{vel_uni}
\end{figure}
\begin{figure}[ht]
\centering
\includegraphics[scale=0.4]{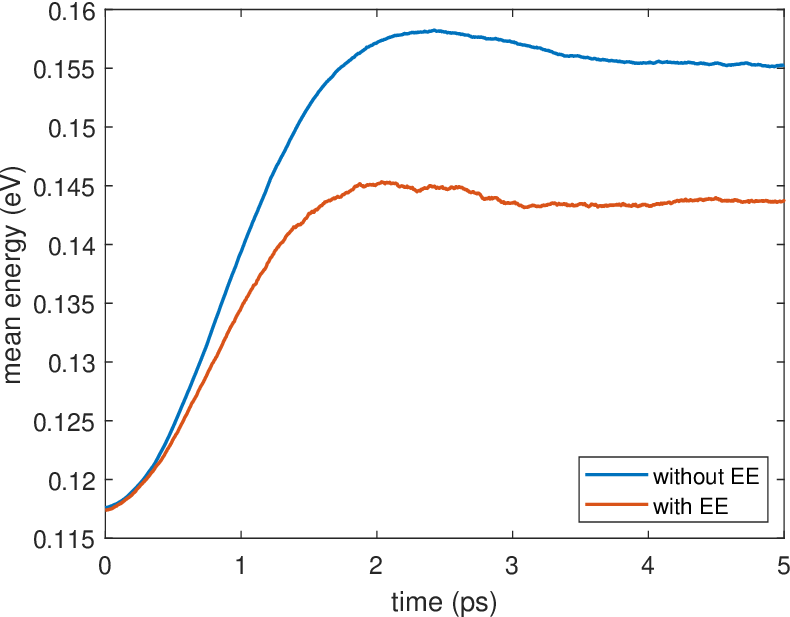}\includegraphics[scale=0.4]{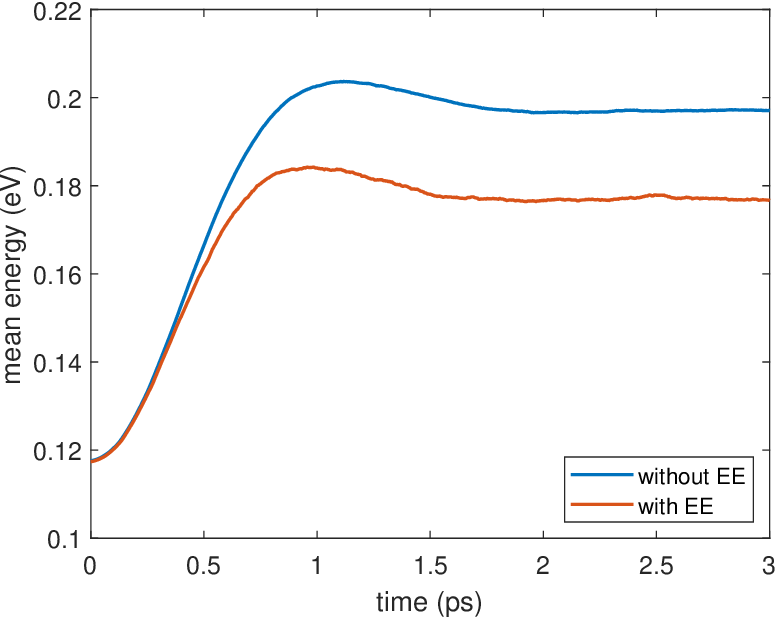}\includegraphics[scale=0.4]{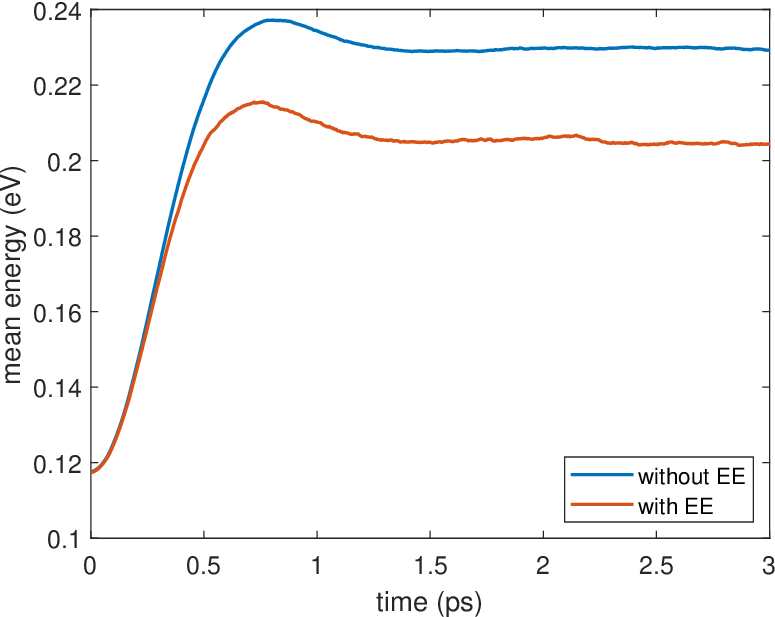}\\
\includegraphics[scale=0.4]{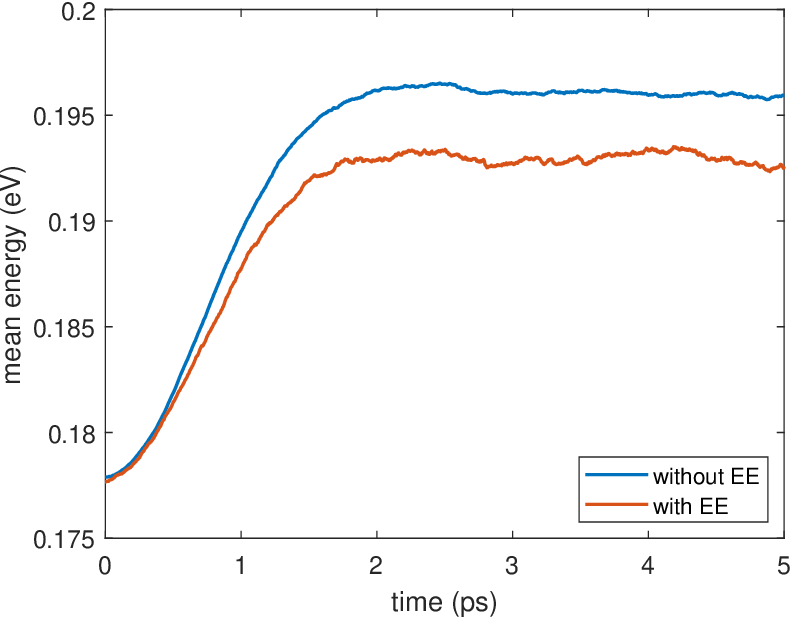}\includegraphics[scale=0.4]{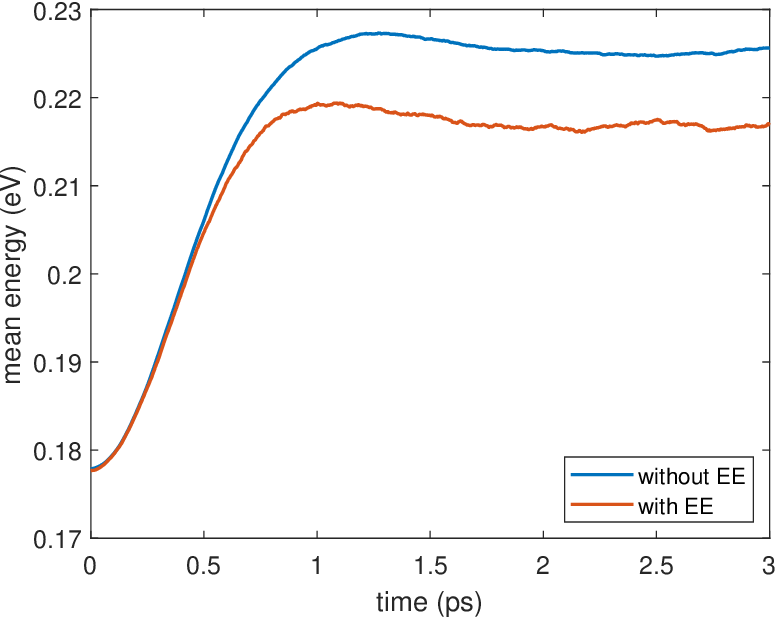}\includegraphics[scale=0.4]{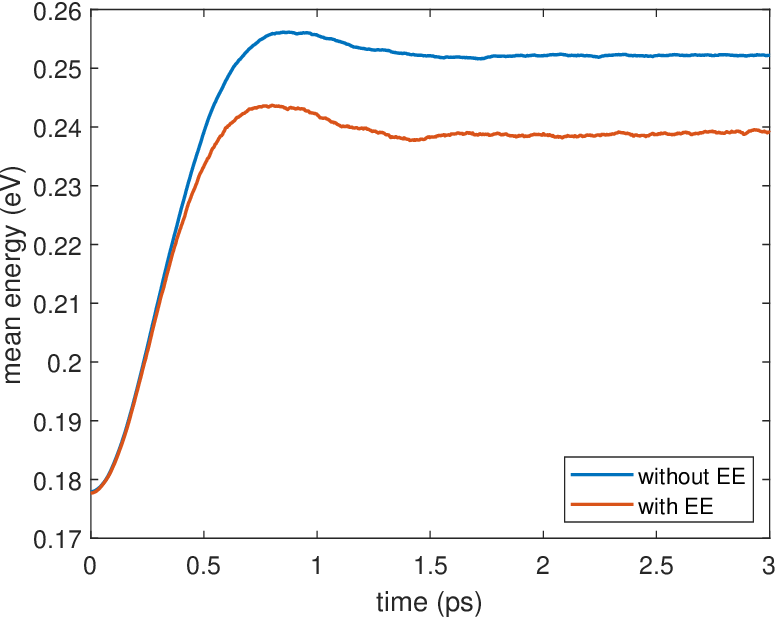}
\caption{Mean energy versus time at Fermi level 0.15 eV (top) and 0.25 eV (bottom) with applied electric fields 1kV/cm (left), 3kV/cm (center), 5kV/cm (right).}
\label{en_uni}
\end{figure}
As appears from the plots, the electron-electron scattering makes the mean velocity lower. 
The difference in the velocity is about 13$\%$ both when $\eps_F$ = 0.15 eV and $\eps_F$ = 0.25 eV. Apparently, a degradation of the velocity could seem in contrast with the property of momentum conservation. However, the group velocity is not proportional to the crystal momentum, as happens in the parabolic band structure, and momentum conservation does not imply velocity conservation. Of course a reduction of the mean velocity implies a non negligible degradation of then current. Moreover, the differences in the velocities influence also the balance equation for the energy. Therefore, despite the property of the electron-electron collision operator of conserving the energy, indeed an effective degradation of the average energy is obtained. Our results are in qualitative agreement with those reported in \cite{ART:Fang} where a standard Monte Carlo approach has been adopted. 

\section*{Conclusion and acknowledgments}
The main analytical properties of the electron-electron scattering have been investigated, in particular the collisional invariants and the equilibria which are more general than the Fermi-Dirac distribution. The scattering rate has been evaluated and used in DSMCs. The obtained results  highlight the importance of the electron-electron scattering for the correct determination of the currents and, therefore, of the characteristic curves. 

As future improvements, the dynamics of phonons could be added, e.g. by using the approaches in \cite{CoRo,CoMaRo,MaRo,Mascali} and some quantum corrections added and simulated for example as in \cite{MuWa}.

The authors thank Prof. A. Majorana for the stimulating discussions.

The authors acknowledge the support from INdAM (GNFM) and from Universit\`{a} degli Studi di Catania, {\it Piano della Ricerca 2020/2022 Linea di intervento 2} ``QICT''. G. Nastasi acknowledges the financial support from the project PON R\&I 2014–2020 ``Asse
IV - Istruzione e ricerca per il recupero - REACT-EU, Azione IV.4 - Dottorati e contratti di ricerca su tematiche dell'innovazione'', project: ``Modellizzazione, simulazione e design di transistori innovativi''.

%


%
%
\clearpage \noindent
\end{document}